\documentclass[11pt]{article}

\usepackage[utf8]{inputenc}
\usepackage[T1]{fontenc}
\usepackage[margin=1in]{geometry}
\usepackage{hyperref}
\usepackage{url}
\usepackage{booktabs}
\usepackage{amsfonts}
\usepackage{amsmath}
\usepackage{amsthm}
\usepackage{nicefrac}
\usepackage{microtype}
\usepackage{xcolor}
\usepackage{natbib}

\newtheorem{informalthm}{Informal Theorem}
\newtheorem{theorem}{Theorem}[section]
\newtheorem{lemma}[theorem]{Lemma}

\newtheorem{cor}[theorem]{Corollary}

\newtheorem{example}{Example}

\DeclareMathOperator*{\Ex}{\mathbb{E}}

\title{Revenue Guarantee of Anonymous Pricing for Mixed Bidders:\\
Bridging Value and Utility Maximizers}

\author{Zhile Jiang \thanks{Department of Computer Science, Aarhus University, Denmark. Email: {\tt zhile@cs.au.dk}.}  \and Stratis Skoulakis \thanks{Department of Computer Science, Aarhus University, Denmark. Email: {\tt stratis@cs.au.dk}.}}

\date{}

\begin{document}

\maketitle

\begin{abstract}
Mechanism design increasingly faces heterogeneous environments containing both traditional utility maximizers and value maximizers, the latter of whom seek to maximize acquired value subject to Return-on-Spend constraints. Designing revenue-optimal mechanisms for such multi-dimensional settings is both computationally and theoretically challenging. To address this complexity, we investigate the revenue guarantees of \textit{Anonymous Pricing} (AP), a simple and practical mechanism, in heterogeneous markets composed of both value and utility maximizers.

By establishing a structural behavioral equivalence between value and utility maximizers, we show that AP, with an appropriately chosen price, achieves a \(1/e\) fraction of the optimal revenue. Our result improves upon the recent \( \frac{1}{2}(1 - 1/e) \) guarantee established by Deng et al.~(2022) for pure value maximizers, while extending it to mixed bidder types (both value and utility maximizers). We additionally establish an upper bound of \(1/2.62\) for AP.

Finally, we demonstrate a counterintuitive phenomenon: competition can reduce revenue with the presence of value maximizers. In particular, running a First-Price Auction with the exact same reserve price as AP can, in the presence of value maximizers, generate lower revenue than AP itself.
\end{abstract}

\section{Introduction}

Motivated by recent trends in online markets, particularly internet advertising with automated bidding agents, the mechanism design community has increasingly focused on value maximizers~\citep{ABM19}. Unlike traditional utility maximizers who seek to maximize quasi-linear utility, value maximizers aim to maximize their expected acquired value subject to a Return on Spend (ROS) constraint---requiring that total value be at least a fixed multiple of total spend. This behavior is pervasive in modern advertising platforms (e.g., Google Ads and Meta Ads), where advertisers deploy ``Target ROAS'' strategies across large-scale, automated campaigns. For this new class of agents, researchers have developed a relatively mature understanding of revenue-maximizing mechanisms: \citet{BDMMZ21}, for instance, demonstrated that mechanisms can fully extract the optimal welfare as revenue, a ``full surplus extraction'' property unavailable against traditional utility maximizers.

In practice, however, value maximizers and utility maximizers may coexist. This heterogeneity changes the mechanism design problem. The designer must now contend with agents whose private information is multi-dimensional, encompassing both valuations and behavioral types, rendering the clean, single-parameter theory inapplicable. Characterizing incentive-compatible mechanisms becomes substantially harder, as monotonicity conditions no longer suffice, and computing a revenue-optimal mechanism is known to be challenging in general~\citep{BDMMZ24}.

These challenges make a compelling case for simple mechanisms. The ``Simple versus Optimal'' paradigm, pioneered by \citet{HR09}, advocates for mechanisms whose structural simplicity delivers three practical benefits simultaneously. First, a low-dimensional parameter space enables the seller to learn near-optimal parameters from data via standard no-regret algorithms, without requiring distributional prior knowledge, a critical advantage in heterogeneous markets where agent types may be unknown. Second, simplicity promotes transparency: participants can readily understand the rules, leading to predictable, stable equilibrium behavior even when agent populations are mixed. Third, and most crucially for our setting, simple mechanisms can sidestep the multi-dimensionality of the problem entirely.

Anonymous Pricing, posting a single uniform price, is the canonical example: because the payment is fixed, the seller's revenue decouples from the agents' private types and depends only on the aggregate selling probability, making AP robust to agent heterogeneity by design. \citet{DMZ22} demonstrated that this robustness extends to the prior-free setting via no-regret learning against pure value maximizers; our work builds on and extends their framework to the heterogeneous case.

\subsection{Our Results and Contribution}

In this work, we focus on anonymous pricing, which is the simplest possible mechanism. We consider the heterogeneous cases where the utility maximizers, the value maximizers, and even their combination, called the hybrid maximizers, co-exist.

\paragraph{Revenue Guarantee of Anonymous Pricing.}
We establish the revenue guarantees of anonymous pricing in these heterogeneous environments. Prior work by \cite{DMZ22} established revenue bounds for value maximizers. Our results extend to heterogeneous populations and improve both the lower and upper bounds of \cite{DMZ22}. Specifically, we improve the approximation ratio interval from $[0.316, 0.479]$ to $[1/e, 1/C^\star]$, where $C^\star \approx 2.62$. We note that our upper bound construction utilizes instances with pure value maximizers.

\begin{informalthm}
    For the heterogeneous environments, the expected revenue of the optimal anonymous pricing mechanism is at least a $1/e$ fraction of the optimal mechanism's expected revenue, and is at most $1/C^\star$ of the optimal mechanism's expected revenue for at least one instance with value maximizers only.
\end{informalthm}

We obtain these bounds by providing a structural connection between value and utility maximizers. We show a reduction from value maximizers to utility maximizers with regular distributions under pricing schemes. This reduction not only improves the previous bounds but also allows for simpler proofs by adapting standard techniques like the ex-ante relaxation framework \citep{AHNPY19} and upper-bound construction with pure utility maximizers \citep{JLTX19}.

\begin{informalthm}
    Under any pricing mechanism, any value maximizer or any hybrid maximizer drawn from a continuous valuation distribution is behaviorally equivalent to a utility maximizer with a transformed, regular valuation distribution.
\end{informalthm}

Building upon our improvement on revenue guarantee of anonymous pricing and existing frameworks \citep{DMZ22}, we show that there exists a prior-free mechanism based on anonymous pricing.

\begin{informalthm}
    In the absence of prior distributional knowledge, there exists a learning-based, prior-free mechanism that extracts at least a $1/e$ fraction of the optimal revenue in the limit.
\end{informalthm}

\paragraph{Revenue Degradation from Competition.}
Given the fact that anonymous pricing provide constant approximation, another question is: Can the mechanism designer always impose an auction using the optimal anonymous price as a reserve without harming overall revenue? In the setting with only utility maximizers, introducing any standard auction satisfying individual rationality will not decrease the revenue, since an agent will always bid at least the reserve price whenever its valuation is above the reserve, ensuring the mechanism extracts at least as much revenue as the posted price. However, we demonstrate that this intuition fails in the presence of value maximizers.
We provide an example showing that this phenomenon can occur even in environments populated solely by value maximizers.

\begin{informalthm}
    Unlike utility-maximizing settings, introducing competition can strictly degrade revenue. There exist instances with pure value maximizers where an anonymous posted price extracts strictly more revenue than a standard first-price auction using the optimal anonymous price as a reserve.
\end{informalthm}

\subsection{Further Related Work}

\paragraph{Mechanism Design with Value Maximizer.}
\citet{ABM19} initiated the study of mechanism design with value maximizers, showing that incentive-compatible auctions cannot extract more than $1/2$ of the optimal welfare. \citet{BDMMZ21} proved that the first-best revenue is achievable when either the valuation or the ROS target is public, but not when neither is public. \citet{DZ21} extended this to a dynamic, prior-independent setting, achieving $\tilde{O}(T^{2/3})$ regret against the first-best benchmark. For the single-agent case where both the valuation and the ROS target are private, \citet{BDMMZ24} provides the revenue-optimal mechanism under the ex-ante ROS constraint. In the heterogeneous setting with coexisting utility and value maximizers, \citet{LZZLYLCW23} study welfare maximization and \citet{BLWXY25} characterize the optimal mechanism for a single mixed bidder. We refer to \citet{ABBBDFGLLM24} for a comprehensive survey.

\paragraph{Revenue Guarantee of Simple Mechanisms.}
Considering environments with pure utility maximizers, \citet{AHNPY19} established the ex-ante relaxation framework and proved that anonymous pricing achieves at least $1/e$ of the optimal revenue for regular utility maximizers.
\cite{JLTX19} provides an upper-bound construction showing that the worst-case revenue guarantee is no better than $C^\star\approx 2.62$. And its tightness is later proven by \cite{JLQTX19}.  There is also a rich literature about revenue guarantees for other simple mechanisms or in multi-item environments, see e.g. \citep{BK96, CHK07, HR09, CHMS10, HHT14, FILS15, CFHOV17, FHL19, JJLZ21, CJK25}. In the case of pure value maximizers, \citet{DMZ22} establishes that sequential posted pricing achieves $1/2$ of the optimal revenue while anonymous pricing can only give $\frac{1}{2}(1-1/e) \approx 0.316$ of the optimal revenue. They also provide an upper bound of approximately $0.479$.

\section{Preliminaries}
\label{sec:pre}
We consider a seller offering a single indivisible item to a set of potential buyers denoted by $\mathcal{N} = \{1,2,\dots,n\}$. Each agent $i \in \mathcal{N}$ is characterized by two private parameters: a deterministic behavioral type $\theta_i \in [0,1]$ and a private valuation $v_i$. The deterministic type $\theta_i$ defines the agent's fundamental objective, while the valuation $v_i$ is a Bayesian variable. We assume $v_i$ is drawn independently from a public, continuous cumulative distribution function $F_i$ with a well-defined probability density function $f_i$ over a non-negative support.

\noindent \textbf{Mechanisms.}
A mechanism $\mathcal{M}$ in this environment is defined by an allocation rule $\mathbf{x}$ and a payment rule $\mathbf{p}$, which map from a reported bid profile $\boldsymbol{\beta} = (\beta_1, \dots, \beta_n)$ to outcomes. Here, a bid $\beta_i = (\hat{\theta}_i, \hat{v}_i)$ is a reported pair of the agent's type and valuation. For any reported profile $\boldsymbol{\beta}$, the allocation rule $\mathbf{x}(\boldsymbol{\beta}) = (x_1(\boldsymbol{\beta}),\dots,x_n(\boldsymbol{\beta}))$ specifies the probability that each buyer $i$ receives the item, subject to the feasibility constraint $x_i(\boldsymbol{\beta}) \in [0,1]$ and $\sum_{i=1}^{n}x_i(\boldsymbol{\beta}) \leq 1$. Simultaneously, the payment rule $\mathbf{p}(\boldsymbol{\beta}) = (p_1(\boldsymbol{\beta}), \dots, p_n(\boldsymbol{\beta}))$ determines the non-negative monetary transfer $p_i(\boldsymbol{\beta}) \geq 0$ that buyer $i$ must effectively pay to the seller.

Assuming all other buyers follow a given bidding strategy profile $\mathbf{s}_{-i}$, we define the interim allocation $x_i(\beta_i) = \Ex_{\mathbf{v}_{-i}}[x_i(\beta_i, \mathbf{s}_{-i}(\mathbf{v}_{-i}))]$ and the interim payment $p_i(\beta_i) = \Ex_{\mathbf{v}_{-i}}[p_i(\beta_i, \mathbf{s}_{-i}(\mathbf{v}_{-i}))]$. These represent the expected probability of winning and the expected payment for buyer $i$ when reporting $\beta_i$. Consequently, the total expected revenue extracted by the seller is entirely governed by this payment rule.

To evaluate the outcome of the mechanism, we consider the revenue of the mechanism when the agents are in a Bayesian Nash equilibrium. Since this paper does not use any property of equilibria, we omit the definition of the Bayesian Nash equilibrium. Readers only need to remember that any Bayesian Nash equilibrium can be represented as a vector $\mathbf{s}$, which is a collection of mappings from valuation to bids. And it will not break the feasibility constraint on the allocation rule. Formally, we defined the revenue as:
\begin{align*}
    \text{REV}_{\mathcal{M}}(\mathbf{s}) = \Ex_{\mathbf{v} \sim \mathbf{F}}\left[\sum_{i=1}^n  p_i(s_i(v_i))\right].
\end{align*}
In this paper, our results focus on simple mechanisms like anonymous pricing. Consequently, we will largely abstract away from the complex, point-wise allocation details of mechanisms. Instead, our analysis will heavily rely on these expected, interim quantities (and the resulting expected revenue), which naturally lend themselves to the ex-ante relaxation techniques we introduce later.

\noindent \textbf{Types of Agents.}
In general, agent $i$ with type $\theta_i$ employs a bidding strategy $s_i(v_i)$ to maximize the following generalized objective:
\begin{align*}
    u_i(s_i) = \Ex_{v_i \sim F_i} [v_i \cdot x_i(s_i(v_i)) - \theta_i \cdot p_i(s_i(v_i))]
\end{align*}
This objective captures the difference between the expected value gained from the mechanism and the expected discounted payment. The deterministic discount factor $\theta_i$ controls how sensitive the agent is to monetary spend. This unified formulation covers three distinct classes of agents:
\begin{itemize}
    \item \textbf{Utility Maximizer ($\theta_i=1$):} The agent seeks to maximize standard quasi-linear utility.
    \item \textbf{Value Maximizer ($\theta_i=0$):} The agent strictly maximizes their expected gained value $\Ex[v_i \cdot x_i]$. As detailed below, this pure volume-maximizing behavior is typically regulated by an overall Return on Spend constraint.
    \item \textbf{Hybrid Maximizer ($0<\theta_i<1$):} Represents an intermediate strategy where the agent prioritizes value accumulation but retains a fractional sensitivity to payment costs. We also assume these agents are subject to the Return on Spend constraint.
\end{itemize}

\noindent \textbf{Return on Spend (ROS) Constraint.}
In practice, pure value maximizers ($\theta_i=0$) do not bid with infinite budgets. Their volume-maximizing behavior is subject to a strict Return on Spend (ROS) constraint. Formally, an ROS constraint requires that the total expected value obtained by the agent must be at least a target multiple, $\gamma_i$, of their total expected payment:
    $$\Ex[v_i \cdot x_i(s_i(v_i))] \geq \gamma_i \cdot \Ex[p_i(s_i(v_i))].$$
This target $\gamma_i\ge1$ can also be the agent's private information, alongside their type $\theta_i$ and valuation $v_i$. A mechanism can theoretically adjust the allocation and the payment based on the reported $\gamma_i$.

One might be tempted to simplify the problem by normalizing valuations or payments to absorb these parameters (e.g., reducing the setting to $\theta_i \in \{0,1\}$ and $\gamma_i = 1$). However, because an agent's true $\theta_i$ and $\gamma_i$ are private information, such type-dependent normalization is technically impossible for the designer to implement. Furthermore, even as a purely analytical tool, normalization fails under anonymous pricing: applying different scaling factors to different agents transforms a single uniform price into varying effective prices, which artificially distorts the aggregate revenue and selling probabilities. Therefore, to correctly evaluate anonymous pricing, we must directly analyze the full continuous spectrum of $\theta_i$ and $\gamma_i$.

\paragraph{Regular Distributions.} For utility maximizers or hybrid maximizers, we use the standard regularity assumption. A valuation distribution $F$ is regular if and only if its virtual valuation $v_F(z)=z-\frac{1-F(z)}{f(z)}$ is non-decreasing, or equivalently, its quantile-revenue curve $R(q)=q\cdot F^{-1}(1-q)$ is concave.

\section{Anonymous Pricing}

In this paper, we focus on the revenue guarantee of the simplest mechanism called \emph{Anonymous Pricing} (AP). The execution of $\mathrm{AP}$ with price $p$ is straightforward: the seller posts a single, uniform price $p \ge 0$ for the item. As agents arrive sequentially, the item is allocated to the first agent whose objective allows them to accept $p$. Upon a successful transaction, the seller collects a payment of exactly $p$, and the mechanism terminates. The behavior of agents is simple under AP, and the seller's expected revenue equals the posted price multiplied by the probability of sale. We denote the revenue of AP with price $p$ as $\mathrm{REV}_{\mathrm{AP}(p)}$.

The main advantage of anonymous pricing in our setting is that its revenue depends only on the ex-ante selling probability, making it less sensitive to the underlying bidder heterogeneity. Because the payment is strictly fixed to $p$, the seller's expected revenue is entirely decoupled from the specific type or valuation of the winning agent. Instead, it depends exclusively on the overall probability that the item is sold to any agent. Consequently, the expected revenue is independent to the agents' arrival sequence or the mechanism's tie-breaking rules.

We observe that the optimal decision-making process for these heterogeneous agents under AP naturally follows a threshold strategy. Specifically, any agent $i$ accepts the price $p$ if and only if their private valuation $v_i$ is at least a specific threshold $\tau_i(p)$, which is uniquely governed by their type $\theta_i$ and valuation distribution $F_i$:
\begin{itemize}
    \item \textbf{Utility Maximizer ($\theta_i=1$):} The agent seeks to guarantee a non-negative quasi-linear utility on every transaction. Thus, their threshold is simply the price itself: $\tau_i(p) = p$.
    \item \textbf{Value Maximizer ($\theta_i=0$):} Driven by volume rather than individual bid utility, this agent accepts the price as long as their expected valuation, conditional on winning, covers the price to satisfy their ROS constraint. Define their threshold $\tau_i(p)$ as the minimum value satisfying  $\Ex_{v_i\sim F_i}[v_i \mid v_i \ge \tau_i(p)] = \gamma_i\cdot p$. If no such value exists (i.e., the price $p$ guarantees a loss), the agent never participates, which is formally denoted as $\tau_i(p) = \infty$.
    \item \textbf{Hybrid Maximizer ($0 < \theta_i < 1$):} Bridging the two extremes, a hybrid agent's optimal strategy is governed by a composite threshold: $\tau_i(p) = \max(\theta_i \cdot p, \alpha_i)$, where $\alpha_i$ is similarly defined as the minimum value such that  $\Ex_{v_i\sim F_i}[v_i \mid v_i \ge \alpha_i] = \gamma_i\cdot p$ (with $\alpha_i = \infty$ if no such value exists). This reflects a threshold considering both the ROS constraint and their fractional sensitivity to the payment.
\end{itemize}

\subsection{Equivalence among different types of agents}

We are now ready to establish the core structural equivalence that allows us to bridge the gap between pure value-maximizing and utility-maximizing behaviors. Intuitively, because the optimal strategy for any agent under anonymous pricing reduces to a threshold rule, the seller's revenue is entirely determined by the probability distribution of this threshold. By constructing a distribution that matches the induced threshold behavior, we can explicitly construct an equivalent cumulative distribution function. If a standard value maximizer were to draw their valuation from this virtual distribution, their purchasing behavior and the revenue they generate would become entirely indistinguishable from the original value maximizer.

\begin{lemma}
    \label{lem:val-equivalence}
    For any value maximizer with a continuous valuation distribution $F$ and a private ROS target $\gamma \ge 1$, there exists a corresponding standard utility maximizer drawing their valuation from a regular distribution $G$. Under any anonymous price $p > 0$, both agents exhibit identically matching purchasing probabilities and generate the exact same expected revenue.
\end{lemma}

\begin{proof}
    Let $v \sim F$ be the valuation of the value maximizer. Given a posted price $p$ and an ROS target $\gamma \ge 1$, the agent employs a threshold strategy: purchase if $v \ge \tau$. The agent's objective is to maximize the expected value $V(\tau) = \int_{\tau}^{\infty} x f(x) dx$, subject to the ROS constraint $V(\tau) \ge \gamma \int_{\tau}^{\infty} p f(x) dx$.

    Dividing the constraint by the purchase probability $1 - F(\tau)$, we obtain the conditional expectation form: $\Ex[v \mid v \ge \tau] \ge \gamma p$. Let us define the conditional expectation function $A(\tau) = \Ex[v \mid v \ge \tau]$. Because the agent seeks to maximize the expected value (which means maximizing the allocation probability $1 - F(\tau)$ by pushing $\tau$ as low as possible), the optimal threshold $\tau(p)$ will bind the constraint tightly. To formally account for cases where $F$ might be flat in certain intervals (meaning $f(x)=0$ and $A(\tau)$ is locally constant), we define the threshold via the infimum:
    \begin{align*}
        \tau(p) = \inf \{ \tau \in \text{supp}(F) \mid A(\tau) \ge \gamma p \}
    \end{align*}
    By definition, the expected revenue from this value maximizer is $R_{\text{VM}}(p) = p \cdot (1 - F(\tau(p)))$.

    Now, consider a standard utility maximizer with valuation $z$ drawn from a constructed CDF $G$. Facing the same price $p$, this agent purchases if $z \ge p$, generating revenue $R_{\text{UM}}(p) = p \cdot (1 - G(p))$.
    To ensure identical purchasing probabilities (and thus identical revenue) for any price $p$, we explicitly construct $G$ such that:
    \begin{align*}
        1 - G(p) = 1 - F(\tau(p)) \implies G(p) = F(\tau(p))
    \end{align*}

    We must now prove that this constructed distribution $G$ is regular, i.e., its virtual valuation function $v_G(z) = z - \frac{1 - G(z)}{g(z)}$ is monotonically non-decreasing.

    First, we differentiate the identity $A(\tau)(1 - F(\tau)) = \int_{\tau}^{\infty} x f(x) dx$ with respect to $\tau$:
    \begin{align*}
        A'(\tau)(1 - F(\tau)) - A(\tau)f(\tau) = -\tau f(\tau) \implies A'(\tau) = \frac{f(\tau)(A(\tau) - \tau)}{1 - F(\tau)}
    \end{align*}
    Notice that $A'(\tau) \ge 0$, confirming $A(\tau)$ is non-decreasing.

    Next, taking $z = p$ as the price variable, the PDF of $G$ is $g(z) = \frac{d}{dz}F(\tau(z)) = f(\tau(z))\tau'(z)$. From the binding condition $A(\tau(z)) = \gamma z$, differentiating both sides yields $A'(\tau(z))\tau'(z) = \gamma$, so $\tau'(z) = \frac{\gamma}{A'(\tau(z))}$.
    Substituting $A'(\tau)$ into this expression gives the density of $G$:
    \begin{align*}
        g(z) = f(\tau(z)) \frac{\gamma (1 - F(\tau(z)))}{f(\tau(z))(A(\tau(z)) - \tau(z))} = \frac{\gamma (1 - G(z))}{A(\tau(z)) - \tau(z)}
    \end{align*}
    The inverse hazard rate of $G$ simplifies:
    \begin{align*}
        \frac{1 - G(z)}{g(z)} = \frac{A(\tau(z)) - \tau(z)}{\gamma}
    \end{align*}
    Substituting $A(\tau(z)) = \gamma z$, we get:
    \begin{align*}
        \frac{1 - G(z)}{g(z)} = z - \frac{\tau(z)}{\gamma}
    \end{align*}
    Finally, substituting this into the virtual valuation formula for $G$, we obtain:
    \begin{align*}
        v_G(z) = z - \left( z - \frac{\tau(z)}{\gamma} \right) = \frac{\tau(z)}{\gamma}
    \end{align*}
    By our generalized inverse construction $\tau(z) = \inf\{\tau \mid A(\tau) \ge \gamma z\}$, as the required threshold target $\gamma z$  increases, the minimum satisfying threshold $\tau(z)$ must also increase (it leaps over any flat regions of $F$). Therefore, $v_G(z) = \tau(z)/\gamma$ is increasing in $z$, completing the proof that $G$ is regular.
\end{proof}

Building upon the structural equivalence established for pure value maximizers, we can naturally extend this analytical framework to hybrid maximizers. A hybrid agent bridges the gap between the two extremes. As discussed earlier, their decision-making is governed by a composite threshold $\max(\theta \cdot p, \alpha)$. Intuitively, this $\max$ operator partitions the price space into two distinct regimes: one where the agent's behavior is bottlenecked by their Return on Spend (ROS) constraint, and another where they are constrained by their fractional sensitivity to the price. As $\theta \to 1$, the hybrid agent's behavior asymptotically approaches that of a pure utility maximizer. To ensure that this utility-driven regime remains theoretically well-behaved, we naturally require the agent's underlying valuation distribution $F$ to be regular. Provided this standard regularity assumption holds, we demonstrate that the equivalence mapping holds robustly even for this generalized hybrid class.

\begin{lemma}
    \label{lem:hybrid-equivalence}
    Consider a hybrid maximizer with type $\theta \in (0, 1)$, a private ROS target $\gamma \ge 1$, and a regular valuation distribution $F$. There exists a corresponding standard utility maximizer drawing their valuation from a regular distribution $G$ such that, under any anonymous price $p > 0$, both agents exhibit identical purchasing probabilities and generate the exact same expected revenue.
\end{lemma}

\begin{proof}
Let $v \sim F$ be the valuation of the hybrid maximizer with type $\theta \in (0,1)$ and ROS target $\gamma \geq 1$. As established in Section~3, the agent employs a threshold strategy $\tau_H(p) = \max(\theta p,\, \alpha(p))$ under anonymous price $p$, where $\alpha(p) = \inf\{x \in \mathrm{supp}(F) \mid A(x) \geq \gamma p\}$ with $A(x) = \mathbb{E}[v \mid v \geq x]$.

We construct the equivalent utility maximizer's distribution by setting $G(p) = F(\tau_H(p))$. Under this construction, a standard utility maximizer drawing valuation $z \sim G$ purchases at price $p$ if and only if $z \geq p$, yielding the purchasing probability $1 - G(p) = 1 - F(\tau_H(p))$, which exactly matches that of the original hybrid maximizer. It remains to show that $G$ is regular, i.e., its virtual valuation function $\varphi_G(p) = p - \frac{1 - G(p)}{g(p)}$ is monotonically non-decreasing. We establish this through the concavity of the revenue curve of $G$.

Since $F$ is non-decreasing and $\tau_H(p) = \max(\theta p,\, \alpha(p))$, we have the pointwise identity
\[
G(p) \;=\; F\bigl(\max(\theta p,\, \alpha(p))\bigr) \;=\; \max\bigl(\underbrace{F(\theta p)}_{G_2(p)},\;\underbrace{F(\alpha(p))}_{G_1(p)}\bigr),
\]
where the second equality follows from the monotonicity of $F$.

Define the \emph{quantile revenue curve} of a distribution $H$ as $\widetilde{R}_H(q) = q \cdot H^{-1}(1 - q)$ for $q \in (0,1)$, where $H^{-1}$ denotes the generalized inverse.  It is a standard result in mechanism design that a distribution $H$ is regular if and only if $\widetilde{R}_H$ is concave (as we mentioned in Section \ref{sec:pre}).

Since $G = \max(G_1, G_2)$ pointwise, the generalized inverse satisfies
\begin{equation}\label{eq:inverse_min}
G^{-1}(y) \;=\; \min\bigl(G_1^{-1}(y),\; G_2^{-1}(y)\bigr).
\end{equation}
To verify \eqref{eq:inverse_min}, observe that
$G^{-1}(y) = \inf\{p : G(p) \geq y\} = \inf\{p : G_1(p) \geq y \text{ or } G_2(p) \geq y\} = \min(G_1^{-1}(y),\, G_2^{-1}(y))$.
Multiplying both sides of \eqref{eq:inverse_min} by $q \geq 0$, the revenue curve of $G$ decomposes as
\begin{equation}\label{eq:rev_min}
\widetilde{R}(q) \;=\; q\cdot G^{-1}(1-q) \;=\; \min\bigl(\widetilde{R}_1(q),\; \widetilde{R}_2(q)\bigr),
\end{equation}
where $\widetilde{R}_i(q) = q \cdot G_i^{-1}(1-q)$ is the quantile revenue curve of $G_i$ for $i = 1, 2$.

We consider two components separately:

\textbf{Component $\widetilde{R}_1$ (ROS-driven regime).} The distribution $G_1(p) = F(\alpha(p))$ arises from the pure value-maximizer construction with ROS target $\gamma$, which is exactly the setting of Lemma~\ref{lem:val-equivalence}. By that lemma, $G_1$ is a regular distribution, and hence $\widetilde{R}_1$ is concave.

\textbf{Component $\widetilde{R}_2$ (price-sensitivity regime).} Since $G_2(p) = F(\theta p)$, its generalized inverse is $G_2^{-1}(y) = F^{-1}(y)/\theta$. Therefore,
    \[
    \widetilde{R}_2(q) \;=\; q \cdot \frac{F^{-1}(1-q)}{\theta} \;=\; \frac{1}{\theta}\,\widetilde{R}_F(q),
    \]
    where $\widetilde{R}_F(q) = q \cdot F^{-1}(1-q)$ is the quantile revenue curve of $F$. By the regularity of $F$, $\widetilde{R}_F$ is concave, and positive scaling preserves concavity, so $\widetilde{R}_2$ is concave.

By \eqref{eq:rev_min}, the revenue curve $\widetilde{R}$ is the pointwise minimum of two concave functions. Since the pointwise infimum of any family of concave functions is concave,\footnote{For completeness: let $\{f_\alpha\}$ be concave functions. For any $\lambda \in [0,1]$ and $x, y$ in the domain,
$\inf_\alpha f_\alpha(\lambda x + (1-\lambda)y) \geq \inf_\alpha [\lambda f_\alpha(x) + (1-\lambda)f_\alpha(y)] \geq \lambda \inf_\alpha f_\alpha(x) + (1-\lambda) \inf_\alpha f_\alpha(y)$.}
$\widetilde{R}$ is concave. Therefore, $G$ is regular.
\end{proof}

\section{Ex-Ante Relaxation of Revenue}

We do not yet know the exact structure of the revenue-optimal mechanism. However, to establish a rigorous theoretical benchmark for optimal revenue, we can employ the \emph{ex-ante} relaxation framework by \citet{AHNPY19}. This technique serves as a powerful analytical tool, providing a robust upper bound by completely decoupling the intertwined allocation dynamics among competing agents.

We can intuitively view this relaxation through the lens of independent decoupled allocations. Consider any mechanism, a feasible allocation will provide each agent $i\in[n]$ an ex-ante probability $q_i$ of getting the item, which equals the expectation of interim allocation $\Ex_{v_i\sim F_i}[x_i(s_i(v_i))]$.
Standard feasibility dictates that the sum of these individual allocation probabilities satisfies $\sum_{i=1}^n q_i \le 1$. The relaxation imagines a modified environment where the seller offers a separate copy of the item to each agent independently, requiring only that the \emph{expected} total number of copies sold across all participants does not exceed one. Under this complete decoupling, the expected revenue extracted from any individual agent $i$ is constrained solely by their own ex-ante allocation probability $q_i$. Consequently, the maximum expected revenue of any valid mechanism is upper-bounded by the following mathematical program over all feasible ex-ante allocation probability vectors $\mathbf{q}$:
\begin{align*}
    \text{Maximize} \quad & \sum_{i=1}^n R_i(q_i) \\
    \text{Subject to} \quad & \sum_{i=1}^n q_i \le 1, \\
    & q_i \ge 0, \quad \forall i \in [n].
\end{align*}

In this decoupled state, the maximum revenue $R_i(q_i)$ extractable from agent $i$ represents the supremum of expected payments over all individually rational mechanisms that guarantee an interim allocation probability of at most $q_i$. We can characterize this single-agent revenue as $R_i(q_i) = q_i \cdot p_i(q_i)$, where $p_i(q_i)$ is formally defined as the highest possible price the seller can set while ensuring a purchasing probability of at least $q_i$ under the specific type and constraints of the agent.

The core challenge lies in characterizing the price-probability mapping $p_i(\cdot)$. For a standard utility maximizer (UM), this is simply the inverse survival function $p_i(q_i) = F_i^{-1}(1 - q_i)$. Our previous technical analysis allows us to bridge this gap universally for any agent driven by Return on Spend (ROS) constraints. By synthesizing the structural equivalences established for both value and hybrid maximizers, we can formalize a price-probability bijection that is perfectly identical to the one for a regular utility maximizer.

\begin{cor}
    \label{cor:bijection}
    For any value maximizer or regular hybrid maximizer with an ROS target $\gamma_i \ge 1$, and for any purchasing probability $q_i \in (0, 1)$, there exists a uniquely determined anonymous price $p_i(q_i)$ that induces this exact probability. Furthermore, this required price is identical to the price induced by their corresponding equivalent regular utility maximizer described in Lemma \ref{lem:val-equivalence} or Lemma \ref{lem:hybrid-equivalence}.
\end{cor}

\begin{proof}
By Lemma \ref{lem:val-equivalence} and \ref{lem:hybrid-equivalence}, the purchasing probability for any value or hybrid maximizer at an arbitrary price $p$ corresponds exactly to the survival function of a constructed regular distribution $S_i(p) = 1 - G_i(p)$. Since $G_i$ is guaranteed to be regular, its probability density function $g_i(p)$ remains strictly positive over its interior support. Consequently, the survival function $S_i(p)$ is strictly monotonically decreasing and continuous across the active price range. Hence $S_i$ is invertible on $(0, 1)$, defining the price-probability map $p_i(q_i) = S_i^{-1}(q_i)$. Since both the original complex agent and the constructed value maximizer share this exact same survival function, the price-probability mapping $p_i(q_i)$ is identical for both.
\end{proof}

This structural equivalence delivers a reduction allowing us to bypass the integral ROS constraints and directly treat the heterogeneous population as a collection of regular utility maximizers. Consequently, our heterogeneous environment can be analyzed through the same ex-ante relaxation used for standard utility-maximizer settings.

\section{Reduction to Pure Utility Maximizer Environment}

In this section, we demonstrate how to formally obtain the lower and upper bounds for the revenue guarantee of anonymous pricing. By utilizing the structural equivalence established in the previous sections, we reduce the heterogeneous environment to one populated solely by regular utility maximizers. This reduction directly yields the following main theorem.

\begin{theorem}
    \label{thm:main}
    For any instance $I$ possibly consisting of value maximizers, regular utility maximizers, and regular hybrid maximizers. Let $p^\star(I)$ be the revenue-optimal anonymous price for instance $I$. We have
    \begin{align*}
        C^\star\leq\sup_{I}\frac{\mathrm{OPT}(I)}{\mathrm{REV}_{\mathrm{AP}(p^\star(I))}}\leq e,
    \end{align*}

    where $C^\star\approx2.62$, and $\mathrm{OPT}(I)$ represents the highest revenue achieved by any mechanism for instance $I$. The first inequality is witnessed by an environment populated entirely by value maximizers.
\end{theorem}

\subsection{The Upper Bound}

Building upon the ex-ante relaxation, we can formally evaluate the revenue guarantee of anonymous pricing against the revenue-optimal mechanism. We frame this worst-case revenue gap as a constrained optimization problem. Let $S_i(p)$ denote the probability that agent $i$ is willing to buy the item at price $p$. The expected revenue extracted by anonymous pricing at any uniform price $p$ is directly formulated as $p \cdot \left(1 - \prod_{i=1}^n (1 - S_i(p))\right)$. By explicitly normalizing this optimal anonymous pricing revenue to be at most $1$, the maximum objective value of our ex-ante relaxation yields the exact upper bound on the revenue gap. We formulate this bounding program as follows:
\begin{align*}
    \text{Maximize} \quad & \sum_{i=1}^n R_i(q_i) \\
    \text{Subject to} \quad & p \cdot \left( 1 - \prod_{i=1}^n \big(1 - S_i(p)\big) \right) \le 1, \quad \forall p > 0, \\
    & \sum_{i=1}^n q_i \le 1, \\
    & q_i \ge 0, \quad \forall i \in [n].
\end{align*}

Let $G_i$ denote the regular valuation distribution constructed in Lemma \ref{lem:val-equivalence} or Lemma \ref{lem:hybrid-equivalence} when agent $i$ is a value maximizer or hybrid maximizer, respectively. For any standard utility maximizer, we simply set $G_i = F_i$. By Corollary \ref{cor:bijection}, we can seamlessly rewrite the preceding programming as follows:
\begin{align*}
    \text{Maximize} \quad & \sum_{i=1}^n q_i \cdot G_i^{-1}(1-q_i) \\
    \text{Subject to} \quad & p \cdot \left( 1 - \prod_{i=1}^n G_i(p) \right) \le 1, \quad \forall p > 0, \\
    & \sum_{i=1}^n q_i \le 1, \\
    & q_i \ge 0, \quad \forall i \in [n],
\end{align*}
where $G_i$ is a regular distribution for all $i$. The optimal objective value of this continuous program is shown to be $e$ by \citet{AHNPY19}. This result  implies the second inequality in our Theorem \ref{thm:main}.

\subsection{The Lower Bound}

\citet{JLTX19} shows that there is an instance with standard utility maximizers such that the optimal anonymous pricing only extracts about $1/C^{\star}\approx 1/2.62$ of the revenue from the optimal posted pricing, therefore, at most $1/C^{\star}$ of the optimal revenue. We now demonstrate that this lower bound for standard utility maximizers translates directly to our pure value maximizer environment. The core of their construction relies on triangular distributions. Since a triangular revenue curve is concave, these distributions are regular. To formally transplant this worst-case result into our environment, we simply need to show that our structural equivalence mapping is surjective. Specifically, any regular utility maximizer can be perfectly simulated by a pure value maximizer.

\begin{cor}
\label{cor:u2v}
    For any regular utility maximizer drawing their valuation from a distribution $G$, there exists a corresponding value maximizer with a target return on spend $\gamma = 1$ and a valid valuation distribution $F$ such that both agents exhibit identical purchasing probabilities under any anonymous price.
\end{cor}

\begin{proof}
    Let $v_G(p) = p - (1-G(p))/g(p)$ be the virtual valuation function of the regular distribution $G$. Because $G$ is regular, $v_G(p)$ is non-decreasing. We define its generalized inverse as $v_G^{-1}(y) = \inf \{ p \mid v_G(p) \ge y \}$. We explicitly construct the target distribution $F$ by matching the survival probabilities: $1 - F(x) = 1 - G(v_G^{-1}(x))$.

    Since $v_G(p)$ is non-decreasing, its generalized inverse is also non-decreasing, guaranteeing that $F(x)$ is a valid cumulative distribution function. To verify this construction, we evaluate the conditional expectation $A(x)=\Ex[v~|~v\geq x]$ at  $x = v_G(p)$. By our construction, $1 - F(v_G(p)) = 1 - G(p)$ and the differentials satisfy $f(v_G(p)) v_G'(p) = g(p)$.

    The expected value above the threshold is $A(v_G(p))=\int_{v_G(p)}^{\infty} t f(t) dt / (1-F(v_G(p))$. Applying the change of variables $t = v_G(s)$, this numerator becomes $\int_p^\infty v_G(s) g(s) ds$. We substitute the definition of virtual valuation directly into the integral:
    \begin{align*}
        \int_p^\infty \left( s - \frac{1-G(s)}{g(s)} \right) g(s) ds = \int_p^\infty \big( s g(s) - (1-G(s)) \big) ds.
    \end{align*}

    Notice that the integrand is exactly the negative derivative of $s(1-G(s))$. Evaluating this boundary yields exactly $p(1-G(p))$. Dividing this result by the survival probability $1-G(p)$, we obtain $A(v_G(p)) = p$. This perfectly satisfies the binding threshold condition $A(\tau) = p$ for a value maximizer operating with $\gamma = 1$.
\end{proof}

By this mapping, the seller faces the exact same worst-case as in the standard setting. And the $1/C^{\star}$ upper bound in pure value maximizer environments leads to the first inequality in Theorem \ref{thm:main}.

\section{Prior-Free Mechanism}

Having shown the existence of anonymous pricing with better approximation ratios, we now turn to the prior-free setting to implement the anonymous pricing. In this practical environment, the seller interacts with the agents repeatedly over multiple rounds but has absolutely no prior knowledge of their valuation distributions. Consequently, the seller must dynamically explore and exploit prices over time.

The work by \cite{DMZ22} investigated this problem, proposing no-regret learning algorithms to dynamically optimize posted prices against value maximizers. By integrating their online learning dynamics with our static equivalence framework, we can immediately improve the theoretical revenue guarantees for Anonymous Pricing. Because the fundamental point-wise interaction between the seller and the value maximizers is strictly isomorphic to that of regular utility maximizers or hybrid maximizers, applying the exact same no-regret learning algorithm directly elevates the asymptotic revenue bound to match our tighter static approximation. We remark that our assumption is that the buyers are less patient than the seller, and we estimate our mechanisms in the long run. In this case, even patient buyers have no incentive to lie.

\begin{theorem}
    \label{thm:priorfree}
    Consider a repeated setting with a mixture of value maximizers, regular utility maximizers, and regular hybrid maximizers, where the seller has no prior knowledge of the valuation distributions. By employing the no-regret learning algorithm for Anonymous Pricing, the seller achieves an expected cumulative revenue of at least
    \begin{align*}
        \frac{1}{e}\cdot \text{OPT}\cdot T-O(T^{2/3}),
    \end{align*}
    recall that $\text{OPT}$ denotes the optimal revenue by the best Bayesian Nash equilibrium among the equilibria of any possible mechanism.
\end{theorem}

\begin{proof}
Due to Theorem \ref{thm:main}, there exists a single price $p^\star$ such that the expected revenue in each round $t \in [T]$ is at least $\frac{1}{e}\cdot\text{OPT}$. Our goal is to show that there is a prior-free mechanism that performs close to the fixed price $p^\star$. Assume without loss of generality that the maximum possible price is $1$. Consider the discretization that puts the possible prices into $K$ possible prices equally, i.e., consider the price space $P=\{1/K,2/K,\dots,1\}$. Therefore, there exists a price $p'\in[p^\star-1/K,p^\star]$ such that the expected revenue in each round is at least
$\frac{1}{e}\cdot\text{OPT}-1/K$.

If we run the optimal no-regret learning algorithm for the finite-armed stochastic bandit feedback environment (for example, \citep{MG17}), the achieved expected revenue is at least
\begin{align*}
    \frac{1}{e}\cdot\text{OPT}\cdot T-T/K-O(\sqrt{KT}).
\end{align*}

Let $K=T^{1/3}$, we have the expected revenue of at least

\begin{align*}
    \frac{1}{e}\cdot\text{OPT}\cdot T-O(T^{2/3}).
\end{align*}
\end{proof}

\section{Revenue Reduction from Introducing Competition}

In standard mechanism design, competition will not harm the revenue. For standard utility maximizers, an auction with a reserve price $r$ trivially dominates a posted pricing mechanism with price $p = r$ in terms of expected revenue. This is because the auction preserves the exact same allocation space (any agent with $v \ge r$ still participates) while extracting a higher payment from the winner. However, when agents are value maximizers with ROS constraints, this dominance no longer holds. Introducing competition can actually decrease the overall selling probability and the overall revenue. To illustrate this phenomenon, we provide the following example.

\begin{example}
    \label{exp:fpa}
    Consider two value maximizers with the same ROS target $\gamma=1$ drawing valuation from the same uniform distribution $U[5/8,1]$. The optimal anonymous pricing with $p^\star=5/6$ provides revenue of $200/243\approx 0.8230$. For the first-price auction with the same price as reserve, there is a Bayesian Nash equilibrium that yields revenue approximately $0.8195$.
\end{example}

 The key intuition is that a value maximizer's ROS budget is a scarce resource, and the two mechanisms allocate it differently. Under anonymous pricing, there is no competitor, so the entire budget is spent on expanding participation, i.e., accepting the price even when $v<p$, which directly increases the seller's selling probability. In the FPA, however, the budget must also finance competitive overbidding against the opponent, i.e., each bidder inflates their bid to win more often, and the opponent does the same. Concretely, the FPA participation threshold is substantially higher than the posted-price acceptance threshold. When the support of the value distribution is narrow, competition between near-identical bidders consumes most of the ROS slack while generating little additional revenue, tipping the balance in favor of the simpler mechanism.

\paragraph{Detailed calculations for Example~\ref{exp:fpa}.}

In this example, the revenue by anonymous price using price $p$ is
\begin{align*}
    REV_{\text{AP}}=p\cdot \left(1-\left(\frac{8(2p-1)-5}{3}\right)^2\right)
\end{align*}

For $p\in[13/16,1]$, the revenue is maximized at $p^\star=5/6$ with corresponding revenue of $200/243$.

Now, consider a First-Price Auction (FPA) with a reserve price $r = p^\star = 5/6$. In an auction, high-value agents must submit competitively higher bids to win. This competitive markup consumes a significant portion of their ROS budget.

To formalize this bidding behavior, we first formulate the value maximizer's objective. The agent seeks to maximize expected value $\mathbb{E}[v \cdot x]$ subject to the global ROS constraint $\mathbb{E}[v \cdot x] \ge \mathbb{E}[p]$. Using the method of Lagrange multipliers, the Lagrangian is $\mathcal{L} = \mathbb{E}[v \cdot x] + \lambda (\mathbb{E}[v \cdot x] - \mathbb{E}[p])$, where $\lambda > 0$ is the shadow price of the binding constraint. Factoring out $\lambda$, we can rewrite the objective as maximizing $\mathbb{E}\big[\frac{1+\lambda}{\lambda} v \cdot x - p\big]$. By defining the endogenous pacing multiplier $\mu = \frac{1+\lambda}{\lambda} > 1$, the optimization problem becomes identical to maximizing $\mathbb{E}[\mu v \cdot x - p]$. This implies that under the global ROS constraint, the value maximizer behaves exactly equivalently to a standard utility maximizer with a scaled true valuation of $\mu v$.

Let $b(v)$ be the symmetric, strictly increasing equilibrium bidding strategy with a participation threshold $\tau$. For an agent with true valuation $v \ge \tau$ in a two-bidder auction with values drawn from $U[a, 1]$ (where $a = 5/8$), the winning probability is $x(v) = F(v) = \frac{v - a}{1 - a}$. If the agent misreports their valuation as $z$, the expected virtual utility is

$$U(z \mid v) = \big(\mu v - b(z)\big) \cdot \frac{z - a}{1 - a}.$$

To find the equilibrium, we apply the first-order condition. Taking the derivative with respect to $z$ and evaluating it at the truth-telling optimum $z = v$ yields:

$$\left. \frac{\partial U}{\partial z} \right|_{z=v} = -b'(v) \cdot \frac{v - a}{1 - a} + \big(\mu v - b(v)\big) \cdot \frac{1}{1 - a} = 0.$$

Multiplying through by $(1-a)$ and rearranging gives

$$\mu v = b'(v)(v - a) + b(v) = \frac{d}{dv}\big[b(v)(v - a)\big].$$

Define the expected payment $P(v) = b(v)(v - a)$. The first-order condition reduces to the differential equation $P'(v) = \mu v$. Integrating from the threshold $\tau$ to $v$ yields

$$P(v) = P(\tau) + \int_{\tau}^{v} \mu \, t \, dt = P(\tau) + \frac{\mu}{2}(v^2 - \tau^2).$$

At the marginal threshold $\tau$, the agent bids the reserve price $r$, so $P(\tau) = r(\tau - a)$. Furthermore, the marginal type $\tau$ must be indifferent between participating and not participating. Therefore, their expected virtual utility must be zero:

$$\big(\mu \tau - b(\tau)\big) \cdot \frac{\tau - a}{1 - a} = 0.$$

Since $\tau > a$, this requires $\mu \tau = r$, which cleanly resolves the multiplier as $\mu = r/\tau$. Substituting this back into the payment function gives

$$P(v) = r(\tau - a) + \frac{r}{2\tau}(v^2 - \tau^2).$$

The corresponding bidding function is

$$b(v) = \frac{P(v)}{v - a} = \frac{r(\tau - a)}{v - a} + \frac{r(v^2 - \tau^2)}{2\tau(v - a)}.$$

To ensure the global ROS constraint is satisfied, the value maximizer must raise its participation threshold $\tau$ until the total expected value exactly equals the total expected payment. Since the common factor $1/(1-a)^2$ appears on both sides and cancels, the binding constraint reduces to

$$\int_{\tau}^{1} v(v - a) \, dv = \int_{\tau}^{1} P(v) \, dv.$$

Evaluating the left-hand side:

$$\text{LHS} = \int_{\tau}^{1} \left(v^2 - \frac{5}{8}v\right) dv = \frac{1 - \tau^3}{3} - \frac{5(1 - \tau^2)}{16}.$$

For the right-hand side, substituting $P(v) = r(\tau - a) + \frac{r}{2\tau}(v^2 - \tau^2)$ and integrating:

$$\text{RHS} = r(\tau - a)(1 - \tau) + \frac{r}{2\tau}\left[\frac{1 - \tau^3}{3} - \tau^2(1 - \tau)\right].$$

This simplifies to

$$\text{RHS} = r(\tau - a)(1 - \tau) + \frac{r(1 - \tau)^2(1 + 2\tau)}{6\tau}.$$

Substituting $a = 5/8$ and $r = 5/6$, the equilibrium threshold $\tau$ is uniquely characterized by the cubic equation

$$48\tau^3 - 77\tau^2 + 58\tau - 20 = 0.$$

This cubic has one real root in the interval $(5/8, 1)$. Numerically, the unique valid solution is $\tau_{\mathrm{FPA}} \approx 0.7309$.

The revenue can be computed using the binding ROS constraint, which implies that each bidder's expected payment equals their expected value conditional on winning. The total revenue is therefore

$$\mathrm{Rev}_{\mathrm{FPA}} = \frac{2}{(1-a)^2} \int_{\tau_{\mathrm{FPA}}}^{1} v(v - a) \, dv \approx 0.8195.$$

\section{Conclusion}

We established a structural equivalence between value maximizers and regular utility maximizers under pricing mechanisms, which allowed us to unify and improve the revenue guarantees of anonymous pricing in heterogeneous markets. Beyond its technical role, this equivalence offers a new conceptual motivation for the regularity assumption: any value maximizer under ROS constraint with an arbitrary continuous valuation distribution naturally induces a regular distribution. In this scenario, regularity arises endogenously rather than being imposed as a modeling convenience.

The most immediate open question is to close the gap between the $1/e$ lower bound and the $1/C^\star$ upper bound. In the pure utility-maximizer setting, the tightness of $C^\star$ was established by \citet{JLQTX19} through a analysis that exploits the structure of Myerson's auction. In our heterogeneous setting, however, the revenue-optimal mechanism lacks a known closed-form characterization, which prevents a direct transplantation of this argument. Progress on this gap likely requires new structural insights into optimal mechanisms for mixed populations.


\bibliographystyle{abbrvnat}
\bibliography{reference}

\end{document}